\documentclass[preprint,12pt]{elsarticle} 
\usepackage[numbers]{natbib}

\usepackage[utf8]{inputenc}
\usepackage[english]{babel}
\usepackage{csquotes}
\usepackage{algorithm}
\usepackage[noend]{algpseudocode}
\algrenewcommand\algorithmicrequire{\textbf{Input:}}
\algrenewcommand\algorithmicensure{\textbf{Output:}}
\algnewcommand{\algorithmicand}{\textbf{ and }}
\algnewcommand{\algorithmicor}{\textbf{ or }}
\algnewcommand{\OR}{\algorithmicor}
\algnewcommand{\AND}{\algorithmicand}

\usepackage{amsmath}
\usepackage{amsfonts}
\usepackage{amssymb}
\usepackage{amsthm}
\usepackage{nicefrac}
\usepackage{algpseudocode,algorithm,algorithmicx}
\usepackage[short]{optidef}
\usepackage{mathrsfs}
\usepackage{colonequals}
\usepackage{diagbox}

\usepackage{tabularx}
\usepackage{makecell}

\usepackage{pgfplots}
\usepackage{pgfplotstable}
\usepackage{tikz}
\usetikzlibrary{arrows, positioning, calc}
\usetikzlibrary{automata,snakes}
\usetikzlibrary{decorations}
\usetikzlibrary{decorations.markings,decorations.pathreplacing}
\usetikzlibrary{shapes.misc}
\usetikzlibrary{backgrounds, intersections}

\newtheorem{theorem}{Theorem}[section]

\newtheorem{corollary}[theorem]{Corollary}
\newtheorem{definition}[theorem]{Definition}
\newtheorem{remark}[theorem]{Remark}
\newtheorem{lemma}[theorem]{Lemma}

\newtheorem{example}[theorem]{Example}

\pgfplotsset{compat=1.14} 

\usepackage{amssymb}

\journal{European Journal of Operational Research}

\begin{document}

\begin{frontmatter}



\title{The \(\{0,1\}\)-knapsack problem with qualitative levels}

\renewcommand*{\thefootnote}{\fnsymbol{footnote}}
\author[a]{Luca E. Schäfer\corref{mycorrespondingauthor}}
\cortext[mycorrespondingauthor]{Corresponding author}
\ead{luca.schaefer@mathematik.uni-kl.de}
\author[a]{Tobias Dietz}
\ead{dietz@mathematik.uni-kl.de}
\author[b]{Maria Barbati}
\ead{maria.barbati@port.ac.uk}
\author[c]{Jos\'{e} Rui Figueira}
\ead{figueira@tecnico.ulisboa.pt}
\author[b,d]{Salvatore Greco}
\ead{salgreco@unict.it}
\author[a]{Stefan Ruzika}
\ead{ruzika@mathematik.uni-kl.de}

\address[a]{Department of Mathematics, Technische Universität Kaiserslautern, 67663 Kaiserslautern, Germany}
\address[b]{Department of Economics and Business, University of Catania, 95129 Catania, Italy}
\address[c]{CEG-IST, Instituto Superior T\'{e}cnico, Universidade de Lisboa, 1049-001 Lisboa, Portugal}
\address[d]{Portsmouth Business School, Centre of Operations Research and Logistics (CORL), \\ University of Portsmouth,  PO1 3DE Portsmouth, United Kingdom}

\begin{abstract}
A variant of the classical knapsack problem is considered in which each item is associated with an integer weight and a qualitative level. We define a dominance relation over the feasible subsets of the given item set and show that this relation defines a preorder. We propose a dynamic programming algorithm to compute the entire set of non-dominated rank cardinality vectors and we state  two greedy algorithms, which efficiently compute a single efficient solution. 
\end{abstract}




\end{frontmatter}



\section{Introduction}
In the knapsack problem, each item has an associated profit and a weight value. The goal is to maximize the overall profit of the selected items under the constraint that the sum of the weights associated with the selected items does not exceed the knapsack capacity \cite{kellerer2004multidimensional, martello2000new}.

The profit and the weight values of the items are usually assumed to be positive values and contribute to the definition of the objective function according to a quantitative evaluation. However, the profit (but also the weight values) of a knapsack problem can have a qualitative content \cite{kasperski2010minmax}. For example, this is the case for a knapsack problem related to urban and territorial planning projects for which the profit expresses the overall environmental sustainability. In this context, it can be reasonable to use ordinal qualitative evaluations (such as “bad, medium, good”), because of the difficulty to assess numerical evaluations, which can be only apparently more precise and, instead, are always arbitrary to some extent.

Indeed, in typical real-world applications the need of introducing qualitative evaluations emerges.  For example, for the selection of research and development projects, it is important to consider not only the financial benefits of the projects but also their sustainability in terms of human and environmental values \cite{vandaele2013sustainable}. Similarly, when selecting infrastructures or urban planning projects  their benefits, in terms of social aspects or environmental impact of the projects, can be hard to quantify \cite{mcgreevy2017complexity}. Furthermore, when  selecting technologies in a company  the adoption of qualitative evaluations can be necessary given the scarcity of information to produce exact quantitative evaluations \cite{saen2006decision}.

Although there exists a vast literature on the knapsack problem,  a formulation that considers the qualitative benefits of the items is still missing. Nevertheless, the need to introduce imprecise or stochastic evaluations has been treated in different works, in which profits and weights associated with the items are different from a generic positive quantitative measure. 
Among those,  fuzzy approaches are quite popular. In this sense, \cite{lin2001using} formulated a fuzzy knapsack problem where the weight of each item is imprecise in the sense that it can be less than or greater than a fixed value. Later, they have proposed genetic algorithms to handle such problems \cite{lin2008solving}. Similarly, trapezoidal fuzzy intervals are used to model imprecision in weights and profits of the items, see \cite{kasperski20070}. Alike, fuzzy triangular numbers were exploited for defining the profits and the weights of the constrained knapsack problem, where  a discount applies when a given quantity of an item is inserted in the knapsack, cf. \cite{changdar2015improved}. 

Likewise, stochastic considerations have also been taken into account for the profits and the weights of items. For example, in \cite{dean2004approximating} the profits are deterministic but the weights  are independent random variables. The items are selected  sequentially, and when inserted in the knapsack their size is determined. Then, the expected value of the profit is maximized. On the same perspective,  \cite{yazidi2018aggregation} formulates the non-linear equality fractional knapsack problem where the profit associated with each item depends on the quantity that is inserted in the knapsack.
Alternatively,  \cite{sbihi2010cooperative} assumed that items could have a different profit value according to different scenarios. They have analyzed a set of feasible solutions for all the possible scenarios. Then, the worst scenario, i.e., the knapsack  maximizing the worst possible outcome, is identified. 

Additionally, an alternative approach is to  define a so-called parametric knapsack problem  where the profits of the items are formulated as affine-linear functions of real-valued variables \cite{giudici2017approximation}. The authors proposed an approximation scheme in order to obtain the optimal solutions of the problem for all the values of the parameter  within a given interval.  

Further, the introduction of stochastic or imprecise weights and profits has been approached in a multicriteria formulation where the selection of items happens according not only to one single criterion but to more criteria, see e.g., \cite{bazgan2009solving}. In this sense, criteria with qualitative benefits are often employed for the selection of a set (portfolio) of items (projects) \cite{salo2011portfolio}. In this context, there are two main approaches.
One is to assign a numerical evaluation to qualitative benefits through elicitation of a value function, cf. \cite{keeney1999identifying}.  Another approach is to assign projects to ordered classes and select the projects in the best classes subject to some given constraints and requirements \cite{stal2011application}.
Therefore, in the former case, after defining the value function, the problem can be reformulated in terms of the classical knapsack problem, while in the latter case, an approach quite different to the classical knapsack formulation is adopted.
For an interesting extension of the former approach see \cite{lahdelma2003ordinal}, which describes a stochastic multicriteria acceptability analysis for group decision making with the aim of ranking and selecting a set of  waste management projects optimizing both quantitative and qualitative criteria.
Later, the evaluations on the criteria  were modeled through score intervals which are assumed to include the true value  and all the non-dominated portfolios were computed by means of a preference programming algorithm \cite{liesio2007preference}. Following that,  some modifications on the interdependencies of the projects were added in \cite{liesio2008robust}.

Besides,  fuzzy numbers were used in a multicriteria context as in \cite{relich2017fuzzy} where a  fuzzy weighted average approach is employed  to select  new product development projects. Alternatively, in  \cite{chang2012fuzzy}, the fuzziness of the profits for the projects are expressed through the definition of a data envelopment analysis approach.
Moreover, the use of ordered weighted averaging  operators is endorsed to deal with the vagueness of the contribution of several research funding programs in \cite{wang2013vague}. 

Finally, a promising approach strongly related to the qualitative optimization knapsack problem considered in this paper, was introduced in \cite{barbati2018optimization} where the evaluations of the items are linked to a set of predefined qualitative benefit levels. More precisely, an item can attain a certain benefit level according to the fact that its evaluation is above or below a given threshold; in this way, the evaluation of the items becomes an ordinal evaluation. Then, the preferred knapsack is obtained by means of a multiobjective optimization problem whose objectives to optimize are the number of items that attain the considered benefit levels.   
The differences with respect to the approach presented in this paper are mainly the following three:
First, while in \cite{barbati2018optimization} a mulitobjective knapsack problem, cf. \cite{lust2012multiobjective}, is considered, we consider in this paper a single objective knpasack problem.  Second, in \cite{barbati2018optimization} a single "best" solution is searched through interaction with the decision maker, while in our approach we look for the entire set of solutions that are optimal with respect to numerical representations preserving the order of the ordinal evaluations.
Third, in \cite{barbati2018optimization}, in the context of interactive multiobjective optimization, constraints can be added representing specific requirements expressing preferences of the decision maker in terms of minimum satisfaction levels on the considered objectives. In this paper, we consider simply the usual capacity constraint (which, of course, do not prevent to extend the approach we are proposing also in case of further constraints beyond the capacity).

From the above literature review, the necessity of dealing with imprecise or vague evaluations is evident. However, the adoption of ordinal evaluations in combinatorial optimization is very rare \cite{schafer2020shortest} and, to the best of our knowledge, a formulation of a knapsack problem with qualitative evaluations has not yet been provided. Our paper addresses this gap and, therefore,  further expands the potential applications of  the knapsack problem.

The paper is organized as follows. 
In Section~\ref{sec:notation}, we define preliminaries and introduce the notation necessary to formulate the model. In Section~\ref{sec:knapsack}, two greedy algorithms computing a single efficient solution and a dynamic programming procedure to list all the non-dominated rank cardinality vectors for the knapsack problem with qualitative levels are proposed. Finally, Section~\ref{sec:conclusion} concludes the paper.

\section{Preliminaries and Notation}\label{sec:notation}
Let \(\mathcal{S}=\{s_1,\ldots,s_n\}\) denote a set of items, let \(\mathcal{L}=\{\ell_1,\ldots,\ell_k\}\) denote a set of qualitative levels with \(\ell_1\prec \ell_2 \prec \dots \prec \ell_k\) and let \(W \in \mathbb{N}\) denote the knapsack capacity. With \(\ell_i \prec \ell_{i+j}, i = 1,\ldots, k-1\) and \(j=1,\ldots,k-i\), we indicate that \(\ell_{i+j}\) is strictly better than \(\ell_{i}\) for all \(j = 1,\ldots,k-i\). Furthermore, let \(w \colon \mathcal{S} \rightarrow \mathbb{N}\) denote a function assigning a weight to each item \(s_i \in \mathcal{S}\) and \(r \colon \mathcal{S} \rightarrow \mathcal{L}\) denote a rank function assigning a qualitative  level to each item \(s_i \in \mathcal{S}\). Finally,  \(\mathcal{S}(W)= \{S' \subseteq \mathcal{S} \mid w(S') \leq  W\}\) is used to denote the set of all  feasible subsets of \(\mathcal{S}\) satisfying the capacity \(W \in \mathbb{N}\), where \(w(S') = \sum_{s \in S'} w(s)\).

For the purpose of item sets with qualitative levels, we recall some definitions of binary relations, cf. \cite{ehrg}.
A binary relation on \(\mathcal{L}\) is a subset \(\mathcal{R}\) on \(\mathcal{L} \times \mathcal{L}\). For any $\ell,\ell^\prime \in \mathcal{L}$, $(\ell,\ell^\prime) \in \mathcal{R}$ can also be denoted as $\ell \mathcal{R} \ell^\prime$.  
\begin{definition}[Properties of a binary relation]
	A binary relation \(\mathcal{R}\) on \(\mathcal{L}\) is called
	\begin{itemize}
		\item[1)] reflexive, if \((\ell,\ell) \in \mathcal{R}\) for all \(\ell \in \mathcal{L}\)
		\item[2)] transitive, if \((\ell_1,\ell_2) \in \mathcal{R}\) and \((\ell_2,\ell_3) \in \mathcal{R} \text{ implies } (\ell_1,\ell_3) \in \mathcal{R}\) for all \(\ell_1,\ell_2,\ell_3 \in \mathcal{L}\)
		\item[3)] antisymmetric, if \((\ell_1,\ell_2) \in \mathcal{R}\) and \((\ell_2,\ell_1) \in \mathcal{R} \text{ implies } \ell_1 = \ell_2\) for all \(\ell_1,\ell_2 \in \mathcal{L}\)
	\end{itemize}
\end{definition}

A binary relation \(\mathcal{R}\) on \(\mathcal{L}\) is called a preorder, if it is reflexive and transitive and it is called a partial order, if it is reflexive, transitive and antisymmetric.

Given a preorder \(\preceq\) on \(\mathcal{L}\), two additional relations can be defined as follows.
\begin{align*}
& \ell_1 \prec \ell_2 :\Leftrightarrow \ell_1 \preceq \ell_2 \text{ and } \ell_2 \not\preceq \ell_1 \text{ (asymmetric part of} \preceq)\\
& \ell_1 \sim \ell_2 :\Leftrightarrow \ell_1 \preceq \ell_2 \text{ and } \ell_2 \preceq \ell_1 \text{ (symmetric part of} \preceq).
\end{align*}

In the following, we use the subsequent definition of a numerical representation. For a survey on preference structures and their numerical representations see \cite{fishburn1999preference}.
\begin{definition}[Numerical representation] Let \(S \subseteq\mathcal{S}\) be a subset of the item set, \(\mathcal{L}\) a set of qualitative levels and let \(r:\mathcal{S}\rightarrow \mathcal{L}\) be a rank function. A function \(v: \mathcal{L} \rightarrow \mathbb{Q}^+\) is called a numerical representation with respect to the rank function \(r\) if 
	\begin{align*}
	&r(s_1) \succ r(s_2) \Leftrightarrow v(r(s_1)) > v(r(s_2)), \text{ for all } s_1,s_2 \in S \text{ and}\\
	&r(s_1) \sim r(s_2) \Leftrightarrow v(r(s_1)) = v(r(s_2)), \text{ for all } s_1,s_2 \in S.
	\end{align*}
	Every numerical representation preserves the order of the rank function. Let \(\mathcal{V}_r\) denote the set of all numerical representations with respect to \(r\).
\end{definition}

\begin{definition}[Rank cardinality function]\label{def:rankfunction}
	Let \(S\subseteq\mathcal{S}\) be a subset of the item set, let \(r:\mathcal{S}\rightarrow~\mathcal{L}\) be a rank function and let \(v:\mathcal{L}\rightarrow \mathbb{Q}^+\) be a numerical representation with \(k\) being the number of qualitative levels. The rank cardinality function \(g_i \colon 2^{\mathcal{S}} \rightarrow \mathbb{N}\) is given by
	\begin{equation*}
	g_i(S) = |\{s \in S: r(s) = \ell_i\}| \text{ for } i = 1,\ldots,k.
	\end{equation*}
	and denotes the number of items in \(S\) with \(\ell_i\) being its qualitative level.
	We call \(g(S) \coloneqq (g_1(S),\ldots,g_k(S))^\top\) the rank cardinality vector of \(S \subseteq \mathcal{S}\).
	Further, for \(S \subseteq \mathcal{S}\), we define \(v(S):= \ell_v \cdot g(S)\), where \(\ell_v := (v(\ell_1),\ldots,v(\ell_k))\), which denotes the total value of \(S\) with respect to the numerical representation \(v\).
\end{definition}

\begin{definition}[Efficiency/ Dominance]\label{def:dominance}
	Let \(S_1, S_2 \in \mathcal{S}(W)\) be feasible subsets of the item set for some \(W \in \mathbb{N}\). Then, 
	\begin{itemize}
		\item[1)] \(S_1\) weakly dominates \(S_2\), denoted by \(S_1 \succeq S_2\), if and only if for every \(v \in \mathcal{V}_r\), it holds that \(v(S_1) \geq v(S_2)\).
		\item[2)] \(g(S_1)\) weakly dominates \(g(S_2)\), denoted by \(g(S_1) \succeq g(S_2)\), if  \(S_1 \succeq S_2\),
		\item[3)]  \(S_1\) dominates \(S_2\), denoted by \(S_1 \succ S_2\), if and only if \(S_1\) weakly dominates \(S_2\) and there exists \(v^* \in \mathcal{V}_r\) such that \(v^*(S_1) > v^*(S_2)\).
		\item[4)] \(S^* \in \mathcal{S}(W)\) is called efficient, if there does not exist some \(S \in\mathcal{S}(W)\) with \(S \succ S^*\),
		\item[5)] \(g(S^*)\) is called non-dominated rank cardinality vector, if \(S^*\) is efficient.
	\end{itemize}
\end{definition}

\begin{definition}[Equivalence]
	Let \(S_1, S_2 \subseteq \mathcal{S}\) be two subsets of the item set. \(S_1 \text{ and } S_2\) are called equivalent if and only if \(v(S_1) = v(S_2)\) for all \(v \in \mathcal{V}_r\).
\end{definition}

\begin{remark}\label{re:dom}
Note that we can rewrite the definition of dominance in the following way. \(S_1\) dominates \(S_2\) if and only if \(S_1\) weakly dominates \(S_2\) and \(S_2\) does not weakly dominate \(S_1\).
\end{remark}

Further, \(S_1\) and \(S_2\) are equivalent, if \(S_1 \succeq S_2\) and \(S_1 \not\succ S_2\).

\begin{lemma}
	The dominance relation \(\succeq\) defined on the set of feasible subsets \(\mathcal{S}(W)\) for some \(W \in \mathbb{N}\)  is a preorder.
\end{lemma}
\begin{proof}
	Obviously \(\succeq\) is reflexive since \(v(S) \geq v(S)\) for every \(v\in \mathcal{V}_r\) and all feasible subsets \(S \in \mathcal{S}(W)\). Further, \(\succeq\) is transitive, since for \(S_1, S_2, S_3\in \mathcal{S}(W)\) with \(v(S_1) \geq v(S_2)\) and \(v(S_2) \geq v(S_3)\) for every \(v\in \mathcal{V}_r\), it holds that \(v(S_1) \geq v(S_3)\) due to Definition \ref{def:rankfunction}. Consequently, \(S_1 \succeq S_3\).
\end{proof}

\section{The \(\{0,1\}\)-knapsack problem with qualitative levels}\label{sec:knapsack}
In this section, we introduce the knapsack problem with qualitative levels. An instance \(I=(\mathcal{S},r,w,W)\) of that problem is given by a set of items \(\mathcal{S}\), a rank function \(r\), a weight function \(w\) and a knapsack capacity \(W\).
Given such an instance \(I\), we aim to find all non-dominated rank cardinality vectors, see Definition \ref{def:dominance}.
Further, we assume the number of qualitative levels \(k\) to be fixed.

\begin{lemma}\label{cor:dom}
	Let \(S_1, S_2 \subseteq \mathcal{S}(W)\) be two feasible subsets of the item set for some \(W \in \mathbb{N}\). Then \(S_1\) weakly dominates \(S_2\), i.e., \(S_1 \succeq S_2\), if and only if \( \sum_{i=j}^k g_i(S_1) \geq\sum_{i=j}^k g_i(S_2)\) for all \(j= 1,\ldots,k\).
\end{lemma}

\begin{proof}
	Let \(S_1, S_2 \subseteq \mathcal{S}(W)\) for some \(W \in \mathbb{N}\) and let \(j^*\in\{1,\ldots,k\}\) with \( \sum_{i=j^*}^k g_i(S_1) < \sum_{i=j^*}^k g_i(S_2)\). Further, we set \(M = 4\cdot|\mathcal{S}|\cdot k\) and define the following numerical representation \(v: \mathcal{L}\rightarrow \mathbb{Q}_+\):	
	\begin{equation}\label{eq:star}
	v(\ell_i) =
	\begin{cases}
	i+M & \text{if } i\geq j^* \\
	i & \text{if } i<j^*  
	\end{cases}
	,\ i=1,\ldots,k.
	\end{equation}
	To obtain that \(S_1\) does not weakly dominate \(S_2\), we observe that
	\begin{equation*}
	\frac{M}{2}  \geq \left(\sum_{i=j^*}^k g_i(S_2) + \sum_{i=1}^{j^*-1} g_i(S_1)\right) k > k  \sum_{i=j^*}^k g_i(S_2) + (j^*-1)  \sum_{i=1}^{j^*-1} g_i(S_1), 
	\end{equation*}
	where the last inequality follows from the fact that \(j^*-1<k\). In particular it holds
	\begin{equation}\label{eq:starstar}
		\frac{M+k}{2} \geq k  \sum\limits_{i=j^*}^k g_i(S_2) + (j^*-1)  \sum\limits_{i=1}^{j^*-1} g_i(S_1).
	\end{equation}
	
	Then, it holds
	\begin{align*}
	v(S_2) & = \sum_{i=1}^k g_i(S_2) v(\ell_i) \\
	& \stackrel{\eqref{eq:star}}{=} \sum_{i=1}^{j^*-1} g_i(S_2)  i + \sum_{i=j^*}^{k} g_i(S_2)  (i+M) \\
	& \geq M  \sum_{i=j^*}^k g_i(S_2) \\
	& \stackrel{\eqref{eq:starstar}}{\geq} M  \sum_{i=j^*}^k g_i(S_2)+k \sum_{i=j^*}^k g_i(S_2) + (j^*-1)  \sum_{i=1}^{j^*-1} g_i(S_1) - \frac{M+k}{2} \\
	& = (M+k)  \left(\sum_{i=j^*}^k g_i(S_2) -\frac{1}{2}\right) + (j^*-1) \sum_{i=1}^{j^*-1} g_i(S_1) \\
	& > (M+k)  \sum_{i=j^*}^k g_i(S_1) + (j^*-1) \sum_{i=1}^{j^*-1} g_i(S_1) \\
	& \geq v(S_1).
	\end{align*}
	Hence, it follows that \(S_1\nsucceq S_2\). Thus, we proved that \(S_1\succeq S_2\) implies that \( \sum_{i=j}^k g_i(S_1) \geq\sum_{i=j}^k g_i(S_2)\) for all \(j= 1,\ldots,k\).
	\newline
	
	For the other direction, let $\widetilde{g}_i(S)$, \(S\subseteq \mathcal{S}(W)\), \(i=1,\ldots,k\), denote the number of items $s \in S$ to which the rank function $r$ assigns a level not smaller than $\ell_i$, i.e., 
	\begin{equation*}
		\widetilde{g}_i(S)=|\{s \in S| r(s) \succeq \ell_i\}|.
	\end{equation*}
	Observe that $g_i(S)=\widetilde{g}_i(S) - \widetilde{g}_{i+1}(S)$  for all $i=1,\ldots,k-1$. Thus, we get
	\begin{equation}
		\begin{aligned}
		v(S)=\sum_{i=1}^kv(\ell_i)g_i(S) & = \sum_{i=1}^{k-1}v(\ell_i)\left(\widetilde{g}_i(S)-\widetilde{g}_{i+1}(S)\right)+v(\ell_k)\widetilde{g}_k(S)\\
		&= v(\ell_1)\widetilde{g}_1(S)+\sum_{i=2}^{k}\left(v(\ell_i) - v(\ell_{i-1})\right)\widetilde{g}_i(S).
		\label{sum}
		\end{aligned}
	\end{equation}

	Consequently, if $\sum_{i=j}^k g_i(S_1) \geq\sum_{i=j}^k g_i(S_2)$ for all $j= 1,\ldots,k$, for some $S_1, S_2 \in \mathcal{S}(W)$ or equivalently, $\widetilde{g}_j(S_1) \geq \widetilde{g}_j(S_2)$ for all $j= 1,\ldots,k$, and using \eqref{sum} and the fact that
	\begin{equation*}
		0 \leq v(\ell_1) \leq \ldots v(\ell_{k-1}) \leq v(\ell_k) \text{ for all } v \in \mathcal{V}_r,
	\end{equation*}
	we get that for all $v \in \mathcal{V}_r$ it holds:
	\begin{gather*}
		v(S_1)=v(\ell_1)\widetilde{g}_1(S_1)+\sum_{i=2}^{k}\left(v(\ell_i) - v(\ell_{i-1})\right)\widetilde{g}_i(S_1)\\
		\geq\\
		v(\ell_1)\widetilde{g}_1(S_2)+\sum_{i=2}^{k}\left(v(\ell_i) - v(\ell_{i-1})\right)\widetilde{g}_i(S_2)=v(S_2)
	\end{gather*}
	Thus, we proved that \(\sum_{i=j}^k g_i(S_1) \geq\sum_{i=j}^k g_i(S_2)\) for all \(j= 1,\ldots,k\) implies $S_1 \succeq S_2$, which concludes the proof.
\end{proof}

Next, we show how to obtain a single efficient solution by running a greedy algorithm. Therefore, we need the following definition of a lexicographically ordered item set.
\begin{definition}[\(r\)-lexicographical order]\label{def:rlexorder}
	Let \(S' \subseteq \mathcal{S}\) be a subset of the item set with \(S' = \{s'_1,\ldots,s'_p\}\) and let \(\pi: \{1,\ldots,p\} \rightarrow \{1,\ldots,p\}\) be a permutation. Then, the lexicographically ordered set \(S'_{r-lex}\) with respect to \(r\) is defined as \(S'_{r-lex} = \{s'_{\pi(1)},\ldots,s'_{\pi(p)}\}\) with \(r(s'_{\pi(1)}) \succeq \dots \succeq r(s'_{\pi(p)})\) and if \(r(s'_{\pi(i)}) = r(s'_{\pi(i+1)})\), then \(w(s'_{\pi(i)})\leq w(s'_{\pi(i+1)})\) for \(i=1,\ldots,p-1\). 
\end{definition}

\begin{algorithm}[H]
	\caption{Greedy Algorithm w.r.t. \(r\)}
	\textbf{Input:} An instance \(I=(\mathcal{S},r,w,W)\)\\
	\textbf{Output:} An efficient solution \(S^* \subseteq \mathcal{S}\)
	\begin{algorithmic}[1]
		\State Sort items \(s_i \in \mathcal{S}\) \(r\)-lexicographically, see Definition \ref{def:rlexorder}
		\State \(S^* \gets \emptyset\)
		\For{$i = 1,\ldots,n$}
		\If{\(w(s_i) \leq W\)}
		\State \(S^* \gets S^* \cup \{s_i\}\)
		\State \(W \gets W - w(s_i)\)
		\EndIf
		\EndFor
		\State \Return \(S^*\)
	\end{algorithmic}
	\label{alg:greedyr}
\end{algorithm}

\begin{theorem}
	The solution \(S^*\) returned by Algorithm \ref{alg:greedyr} is efficient.
\end{theorem}
\begin{proof}
	Let \(S' \in \mathcal{S}(W)\) be an arbitrary feasible subset of \(\mathcal{S}\). We show that \(S' \not\succeq S^*\). Therefore, we distinguish two cases.
	
	\textbf{Case 1:} \(g(S') = g(S^*)\).
	
	It follows that \(v(S') = v(S^*)\) for all \(v \in \mathcal{V}_r\) and nothing remains to show.
	
	\textbf{Case 2:} \(g(S') \neq g(S^*)\).
	
	Let \(j^*\) be maximal such that \(g_{j^*}(S')\neq g_{j^*}(S^*)\). Due to construction of Algorithm \ref{alg:greedyr}, it holds that \(g_i(S') = g_i(S^*)\) and \(w(S'_i)\geq w(S^*_i)\) for all \(i = j^*+1,\ldots,k\), where \(S'_i := \{s\in S'\mid r(s)=\ell_i\}\). Note that \(S^*_i\) is analogously defined. Consequently, it follows that 
	\begin{equation}\label{eq:greedyr}
		g_{j^*}(S')<g_{j^*}(S^*).
	\end{equation}
	Next, we define a numerical representation \(v:\mathcal{L}\rightarrow\mathbb{Q}_+\) as follows.
	\begin{equation}\label{eq:numrep}
	v(\ell_i) =
	\begin{cases}
	\frac{1}{2^{j^*-i}} & \text{if } i<j^* \\
	n & \text{if } i = j^*\\
	n+i & \text{if } i>j^*  
	\end{cases}
	\end{equation}
	It follows:
	
	\begin{align*}
	v(S')&=\sum_{i=1}^{k} g_i(S') v(\ell_i) \\
	& \stackrel{\eqref{eq:numrep}}{=} \sum_{i=1}^{j^*-1} g_i(S') \frac{1}{2^{j^*-i}} +  g_{j^*}(S')  n + \sum_{i=j^*+1}^{k} g_i(S') (n+i) \\
	&< n(g_{j^*}(S')+1) + \sum_{i=j^*+1}^{k} g_i(S') (n+i)\\
	& = n(g_{j^*}(S')+1) + \sum_{i=j^*+1}^{k} g_i(S^*) (n+i) \\
	& \stackrel{\eqref{eq:greedyr}}{\leq} n\cdot g_{j^*}(S^*) + \sum_{i=j^*+1}^{k} g_i(S^*) (n+i) \\
	& \leq \sum_{i=1}^{j^*-1} g_i(S^*) \frac{1}{2^{j^*-i}} + n g_{j^*}(S^*) + \sum_{i=j^*+1}^{k} g_i(S^*) (n+i) \\
	& = v(S^*) 
	\end{align*}
	Consequently, it is \(S' \not\succeq S^*\), which concludes the proof.
\end{proof}

\begin{theorem}
	Algorithm \ref{alg:greedyr} runs in \(\mathcal{O}(n\log n)\).
\end{theorem}
\begin{proof}
	The sorting of the items can be done in \(\mathcal{O}(n\log n)\) time.
	The amount of work in the \(\mathbf{for}\)-loop of Algorithm \ref{alg:greedyr} is in \(\mathcal{O}(n)\), since only constant time operations are performed within each iteration. Thus, the running time follows.
\end{proof}

\begin{example}\label{ex:knapsack}
	To illustrate Algorithm \ref{alg:greedyr}, consider the following knapsack instance \(I =(\mathcal{S},r,w,W)\) with \(\mathcal{S}=\{1,2,3,4\}\) and \(W=6\). Table \ref{tab:rgreedy} shows the different items with their corresponding weights and qualitative levels.
	\begin{figure}[H]
		\begin{table}[H]
			\begin{center}
				\begin{tabular}{ c|c|c }	
					item & \(w\) & \(r\) \\ \hline
					1 & 1 & \(\ell_1\) \\ 
					2 & 2 & \(\ell_2\) \\
					3 & 3 & \(\ell_3\)\\
					4 & 4 & \(\ell_4\) 
				\end{tabular}	
			\end{center}
			\caption{Example of Algorithm \ref{alg:greedyr}}
			\label{tab:rgreedy}	
		\end{table}
	\end{figure}
	Due to Definition \ref{def:rlexorder}, it holds that \(\mathcal{S}_{r-lex} = \{4,3,2,1\}\). Consequently, the solution \(S^*\) returned by Algorithm \ref{alg:greedyr} contains the fourth and the second item, i.e., \(S^*=\{2,4\}\). One can check that \(S^*\) is efficient, see Example \ref{ex:exactKnapsack}.
\end{example}

\begin{definition}[\(w\)-lexicographical order]\label{def:wlexorder}
	Let \(S' \subseteq \mathcal{S}\) be a subset of the item set with \(S' = \{s'_1,\ldots,s'_p\}\) and let \(\pi: \{1,\ldots,p\} \rightarrow \{1,\ldots,p\}\) be a permutation. Then the lexicographically ordered set \(S'_{w-lex}\) with respect to \(w\) is defined as \(S'_{w-lex} = \{s'_{\pi(1)},\ldots,s'_{\pi(p)}\}\) with \(w(s'_{\pi(1)}) \leq \dots \leq w(s'_{\pi(p)})\) and if \(w(s'_{\pi(i)}) = w(s'_{\pi(i+1)})\), then \(r(s'_{\pi(i)}) \succeq r(s'_{\pi(i+1)})\) for \(i = 1,\ldots,p-1\). 
\end{definition}

\begin{algorithm}
	\caption{Greedy Algorithm w.r.t. \(w\)}
	\textbf{Input:} An instance \(I=(\mathcal{S},r,w,W)\)\\
	\textbf{Output:} An efficient solution \(S^* \subseteq \mathcal{S}\)
	\begin{algorithmic}[1]
		\State Sort items \(s_i \in \mathcal{S}\) \(w\)-lexicographically, see Definition \ref{def:wlexorder}
		\State \(S^* \gets \emptyset\)
		\For{$i = 1,\ldots,n$}
		\If{\(w(s_i) \leq W\)}
		\State \(S^* \gets S^* \cup \{s_i\}\)
		\State \(W \gets W - w(s_i)\)
		\EndIf
		\EndFor
		\State \Return \(S^*\)
		
	\end{algorithmic}
	\label{alg:wgreedy}
\end{algorithm}

\begin{theorem}
	The solution \(S^*\) returned by Algorithm \ref{alg:wgreedy} is efficient, if \(w(S^*)=W\).
\end{theorem}
\begin{proof}
	Let \(S^*\) denote the solution returned by Algorithm \ref{alg:wgreedy} with \(w(S^*)=W\). It holds that \(\sum_{i=1}^{k}g_i(S^*)\geq\sum_{i=1}^{k}g_i(S')\) for all \(S'\in\mathcal{S}(W)\). If \(\sum_{i=1}^{k}g_i(S^*)>\sum_{i=1}^{k}g_i(S')\), then there is nothing to show. 
	
	In case of \(\sum_{i=1}^{k}g_i(S^*)=\sum_{i=1}^{k}g_i(S')\), it holds that \(w(S')\geq w(S^*)\). If \(w(S')>w(S^*)=W\), we know that \(S' \notin \mathcal{S}(W)\). In the case of \(w(S')=w(S^*)\), we know due to construction of the algorithm that either \(S'\) and \(S^*\) are both efficient or that \(S^*\) dominates \(S'\). Consequently, it follows that \(S^*\) is efficient.
\end{proof}

\begin{example}\label{ex:greedyw}
	Let \(I=(\mathcal{S},r,w,W)\) denote the same instance as in Example~\ref{ex:knapsack}. It holds that \(\mathcal{S}_{w-lex} = \{1,2,3,4\}\). Consequently, the solution returned by Algorithm \ref{alg:wgreedy} contains the first three items, i.e., \(S^*=\{1,2,3\}\), which can be verified to be efficient, see Example \ref{ex:exactKnapsack}. Further, it holds that \(w(S^*) = 6 = W.\)
\end{example}

\begin{example}
	To show that the solution \(S^*\) after \(n\) iterations of Algorithm~\ref{alg:wgreedy} does not have to be efficient in case of \(w(S^*) < W\), consider the following example. Let \(\mathcal{S}=\{1,2\}\) and \(W=3\). Table \ref{tab:wgreedy} shows the different items with their corresponding weights and qualitative levels.
	\begin{figure}[H]
		\begin{table}[H]
			\begin{center}
				\begin{tabular}{ c|c|c }	
					item & \(w\) & \(r\) \\ \hline
					1 & 2 & \(\ell_1\) \\ 
					2 & 3 & \(\ell_2\) \\ 
				\end{tabular}	
			\end{center}
			\caption{Example of Algorithm \ref{alg:wgreedy} with \(w(S^*)<W\)}
			\label{tab:wgreedy}	
	\end{table}
\end{figure}
	The solution returned by Algorithm \ref{alg:wgreedy} contains only the first item, i.e., \(S^*=\{1\}\). However, note that \(S'=\{2\}\) dominates \(S^*\), i.e., \(S'\succeq S^*\), since \(\sum_{i=j}^{2}g_i(S')\geq\sum_{i=j}^{2}g_i(S^*)\) for all \(j=1,2\). 
\end{example}

\begin{corollary}
	Algorithm \ref{alg:wgreedy} runs in \(\mathcal{O}(n\log n)\).\qed
\end{corollary}

Next, we present a dynamic programming algorithm for the knapsack problem with qualitative levels that computes the entire set of non-dominated rank cardinality vectors. Again, let \(I = (\mathcal{S},r,w,W)\) denote an instance of our problem. For all \(i \in \{0,1,\ldots,|\mathcal{S}|\}\), and for all \(x \in \{0,1,\ldots,W\}\), we introduce label sets \(\mathcal{L}_{i,x}\) referring to those non-dominated rank cardinality vectors that use only the first \(i\) items with a total size smaller or equal to \(x\). Initially, we set \(\mathcal{L}_{0,x}\) to be equal to the empty set, i.e., \(\mathcal{L}_{0,x} = \emptyset\) for all \(x \in \{0,1,\ldots,W\}\). Next, we compute \(\mathcal{L}_{i,x}\) for all \(i \in \{1,\ldots,|\mathcal{S}|\}\) and for all \(x \in \{0,1,\ldots,W\}\) by using the following procedure:
\begin{equation*}
	\mathcal{L}_{i,x} = \max_{\succeq} \left\{\mathcal{L}_{i-1,x} \cup (\ell_i \oplus \mathcal{L}_{i-1,x-w_i})\right\},
\end{equation*}
where \(\ell_i \oplus \mathcal{L}_{i-1,x-w_i} := \{\ell_i\cup L\mid L \in\mathcal{L}_{i-1,x-w_i}\}\). Of course, if \(w(s_i)>x\), we set \(\mathcal{L}_{i,x} = \mathcal{L}_{i-1,x}\). 

\begin{algorithm}
	\caption{Exact Algorithm}
	\textbf{Input:} An instance \(I=(\mathcal{S},r,w,W)\)\\
	\textbf{Output:} All non-dominated rank cardinality vectors
	\begin{algorithmic}[1]
		\For{\(x=0,1,\ldots,W\)}
			\State \(\mathcal{L}_{0,x} \gets \emptyset\)
		\EndFor
		
		\For{\(i=1,\ldots,n\)}
			\For{\(x=0,1,\ldots,W\)}
				\State \(\mathcal{L}_{i,x} \gets \max_{\succeq} \left\{\mathcal{L}_{i-1,x} \cup (\ell_i\oplus\mathcal{L}_{i-1,x-w_i})\right\}\)
			\EndFor
		\EndFor
		\State \Return \(\mathcal{L}_{n,W}\)
	\end{algorithmic}
	\label{alg:exact}
\end{algorithm}

\begin{theorem}
	Algorithm \ref{alg:exact} correctly computes the set of non-dominated rank cardinality vectors.
\end{theorem}
\begin{proof}
	Follows immediately by induction over the number of iterations \(i\).
\end{proof}

\begin{theorem}[see \cite{schafer2020shortest}]\label{thm:bound}
	Throughout the execution of Algorithm \ref{alg:exact} the number of labels in \(\mathcal{L}_{i,x}\) is polynomially bounded for all \(i \in \{1,\ldots,|\mathcal{S}|\}\) and for all \(x \in \{0,1,\ldots,W\}\).
\end{theorem}
\begin{proof}
		The number of different rank cardinality vectors corresponding to item sets with exactly \(j\) elements is the same as the number of possibilities to pick \(j\) elements from a set of \(k\) elements with replacement and without order, i.e., this equals \({k+j-1\choose j}\). Note that this represents an upper bound for the number of rank cardinality vectors in \(\mathcal{L}_{i,x}\) corresponding to item sets of size \(j\) for all \(x \in \{0,\ldots,W\}\). A non-dominated rank cardinality vector in \(\mathcal{L}_{i,x}\) contains at most \(i\) items. Summing up over all possible sizes of item sets results in
		\begin{align*}
		\sum_{j=1}^{i} {k+j-1\choose j} & = {k+i\choose i} - 1 = \frac{(i+1) \ldots  (i+k)}{k!} -1 \in O(i^k).
		\end{align*}
\end{proof}


\begin{theorem}
	Algorithm \ref{alg:exact} runs in \(\mathcal{O}(n^{2k+1}W)\).
\end{theorem}
\begin{proof}
	Clearly, the first \(\mathbf{for}\)-loop is in \(\mathcal{O}(W)\). The nested \(\mathbf{for}\)-loop indicates that their are \(nW\) subproblems that have to be solved. As the algorithm proceeds (increasing \(i\)), the worst-case size of the label sets also increases. The amount of work for the \(i\)-th iteration is in \(\mathcal{O}(i^{2k}W)\), since two label sets of size bounded by \(\mathcal{O}(i^k)\), see Theorem \ref{thm:bound}, have to be searched for non-dominance. Consequently, the overall amount of work is in \(\mathcal{O}(n^{2k+1}W)\), which concludes the proof.
\end{proof}


\begin{example}\label{ex:exactKnapsack}
	Let \(I=(\mathcal{S},r,w,W)\) be the same instance as in Examples \ref{ex:knapsack} and \ref{ex:greedyw}. Table \ref{tab:exact} shows the result of the dynamic programming algorithm, see Algorithm \ref{alg:exact}. The table has to be read from the bottom left to the top right.
	\begin{figure}[H]
		\begin{table}[H]
			\begin{center}
				\begin{tabular}{ c|c|c|c|c|c }
					6 & \(\emptyset\) & \(\{\{\ell_1\}\}\) & \(\{\{\ell_1,\ell_2\}\}\) & \(\{\{\ell_1,\ell_2,\ell_3\}\}\) & \(\{\{\ell_1,\ell_2,\ell_3\}, \{\ell_2,\ell_4\}\}\) \\ \hline
					5 & \(\emptyset\) & \(\{\{\ell_1\}\}\) & \(\{\{\ell_1,\ell_2\}\}\) & \(\{\{\ell_2,\ell_3\}\}\) & \(\{\{\ell_2,\ell_3\},\{\ell_1,\ell_4\}\}\)\\ \hline
					4 & \(\emptyset\) & \(\{\{\ell_1\}\}\) & \(\{\{\ell_1,\ell_2\}\}\) & \(\{\{\ell_1,\ell_3\}\}\) & \(\{\{\ell_1,\ell_3\},\{\ell_4\}\}\)\\ \hline
					3 & \(\emptyset\) & \(\{\{\ell_1\}\}\) & \(\{\{\ell_1,\ell_2\}\}\) & \(\{\{\ell_1,\ell_2\},\{\ell_3\}\}\) & \(\{\{\ell_1,\ell_2\},\{\ell_3\}\}\)\\ \hline
					2 & \(\emptyset\) & \(\{\{\ell_1\}\}\) & \(\{\{\ell_2\}\}\) & \(\{\{\ell_2\}\}\) & \(\{\{\ell_2\}\}\)\\ \hline	
					1 & \(\emptyset\) & \(\{\{\ell_1\}\}\) & \(\{\{\ell_1\}\}\) & \(\{\{\ell_1\}\}\) & \(\{\{\ell_1\}\}\)\\ \hline
					0 & \(\emptyset\) & \(\emptyset\) & \(\emptyset\) & \(\emptyset\) & \(\emptyset\)\\ \hline
					\diagbox[width=3em, dir=SW]{\(x\)}{\(i\)}&
					0 & 1 & 2 & 3 & 4 \\
				\end{tabular}	
			\end{center}
			\caption{Example of Algorithm \ref{alg:exact}}
			\label{tab:exact}	
		\end{table}
	Note that \(\mathcal{L}_{4,6}\) refers to the set of non-dominated rank cardinality vectors of our knapsack instance as described in Example \ref{ex:knapsack}. Thus, \(S_1=\{1,2,3\}\) and \(S_2=\{2,4\}\) denote the solutions of our problem.
	\end{figure}
	
\end{example}

\section{Conclusion}\label{sec:conclusion}
In this paper, we investigated the \(\{0,1\}\)-knapsack problem with qualitative levels. We introduced a concept of dominance for item sets with qualitative levels. We showed that this concept defines a preorder on the set of feasible subsets of a given item set. We proved that the number of non-dominated rank cardinality vectors is polynomially bounded for a fixed number of qualitative levels. We provided a dynamic programming algorithm, which computes the entire set of non-dominated rank cardinality vectors in pseudo-polynomial time and two greedy algorithms, which efficiently compute a single efficient solution.

\section*{Acknowledgements}
This work was partially supported by the Bundesministerium für Bildung und Forschung (BMBF) under Grant No. 13N14561 and by Deutscher Akademischer Austauschdienst (DAAD) under Grant No. 57518713.
José Rui Figueira also acknowledges the support of national funds through FCT research DOME Project with reference PTDC/CCI-COM/31198/2017.
\newpage
\bibliographystyle{apalike}
\bibliography{biblio}

%
%
%
\end{document}